\documentclass[11pt]{article}

\usepackage[utf8]{inputenc}
\usepackage[T1]{fontenc}
\usepackage[english]{babel}
\usepackage[letterpaper]{geometry}
\usepackage{mathptmx}
\usepackage{amsmath,amssymb,amsthm,latexsym}
\usepackage{booktabs}
\usepackage{graphicx}
\usepackage{xcolor}
\usepackage{array}
\usepackage[font=normalsize,skip=2pt]{caption}
\usepackage{titlesec}
\usepackage{setspace}
\usepackage[round,authoryear]{natbib}
\usepackage[hypertexnames=false,linktocpage=true]{hyperref}
\hypersetup{colorlinks=true,allcolors=black,bookmarksnumbered=true,pdfview=FitB,pdfauthor={Serhii Zabolotnii},pdftitle={A Transferability Criterion for Null-Optimized Variance Reduction in Cumulant-Based Error-Independence Testing}}

\titleformat*{\section}{\normalsize\bfseries}
\titleformat*{\subsection}{\normalsize\bfseries}
\titleformat*{\subsubsection}{\normalsize\bfseries}
\titleformat*{\paragraph}{\normalsize\bfseries}
\titleformat*{\subparagraph}{\normalsize\bfseries}
\titleformat{\section}{\sffamily\normalfont\normalsize\bfseries}{\thesection.}{1em}{}
\titleformat{\subsection}{\sffamily\slshape\normalfont\bfseries}{\thesubsection.}{1em}{}
\titleformat{\subsubsection}{\sffamily\slshape\normalfont\bfseries}{\thesubsubsection.}{1em}{}
\titlespacing*{\section}{0pt}{18pt}{6pt}
\titlespacing*{\subsection}{0pt}{18pt}{6pt}
\titlespacing*{\subsubsection}{0pt}{18pt}{6pt}
\newtheorem{theorem}{Theorem}

\newtheorem{proposition}[theorem]{Proposition}
\newtheorem{corollary}[theorem]{Corollary}
\newtheorem{remark}{Remark}

\newcommand{\E}{\mathbb{E}}
\newcommand{\Var}{\operatorname{Var}}
\newcommand{\Cov}{\operatorname{Cov}}

\newcommand{\DcPMM}{\widehat{\Delta c}_3^{\,\text{PMM2}}}
\newcommand{\DcNaive}{\widehat{\Delta c}_3^{\,\text{naive}}}

\begin{document}

\begin{center}
{\Large\bfseries A Transferability Criterion for Null-Optimized Variance Reduction in Cumulant-Based Error-Independence Testing\par}
\vspace{12pt}
{\normalsize Serhii Zabolotnii\(^{1,2,3}\)\par}
\vspace{4pt}
{\small \(^{1}\)Cherkasy State Business College, Cherkasy 18028, Ukraine\par}
{\small \(^{2}\)State Scientific Research Institute of Armament and Military Equipment Testing and Certification, Cherkasy, Ukraine\par}
{\small \(^{3}\)Uzhhorod National University, Uzhhorod, Ukraine\par}
\vspace{4pt}
{\small ORCID: \url{https://orcid.org/0000-0003-0242-2234}\par}
\vspace{4pt}
{\normalsize Corresponding author: \texttt{zabolotnii.serhii@csbc.edu.ua}\par}
\end{center}

\vspace{18pt}

\textbf{Abstract}\vspace{8pt}\\
Control-variate and polynomial-maximization (PMM) estimators are optimized at a
single fixed distribution, yet they are increasingly proposed to strengthen
hypothesis tests, which decide between two regions of a parameter family. We give
a closed-form criterion for when this transfer succeeds. For an $H_0$-centered
augmentation of a target moment statistic with null-optimized weight vector
$\mathbf{K}_0$, the alternative-side expectation equals the target plus
${\mathbf{K}_0}^{\top}\boldsymbol{\mu}_{a,H_1}$, where
\(\boldsymbol{\mu}_{a,H_1}\) is the alternative-side mean of the augmenting
basis. Null-variance reduction therefore transfers without bias only under the
orthogonality condition
${\mathbf{K}_0}^{\top}\boldsymbol{\mu}_{a,H_1}=0$; requiring each augmenting
function to remain mean-zero is sufficient but not necessary. We instantiate the
criterion on the recently proposed Wiedermann--Shi third-order cumulant test for
measurement-error independence. A second-order PMM correction is unbiased and
lower-variance under the null (relative efficiency $\geq 1$ in all 36
conditions; aggregated mean ARE values 1.23--5.16; Type-I 0.04--0.09), yet
provably inconsistent under the alternative: the antisymmetric polynomial
auxiliaries acquire nonzero means, attenuating the target by a closed-form factor
and costing 7--52 percentage points of power, worst where the test is strongest
and worsening under heavy tails. A fourth-order variant reduces variance (ratio
1.127) but fails a nuisance guard (rejection 0.295 versus 0.10). We derive a
reusable alternative-consistency acceptance gate for variance-reduced test
statistics.\vspace{18pt}

 \vspace{12pt}


\section{Introduction}
\label{sec:intro}

Tests of error independence sit at the intersection of psychometric
measurement theory, higher-order moment statistics, and efficient
simulation-based inference. In classical and congeneric measurement
models, reliability claims depend not only on the observed covariance
matrix but also on assumptions about the structure of measurement
errors \citep{cronbach1951alpha,lord1968statistical,
joreskog1971congeneric,mcdonald1999test,raykov1997composite,
revelle2009coefficients}. Design-driven or substantively meaningful
correlated residuals are therefore not a cosmetic modelling detail; if
left unmodelled, they can change the interpretation of latent-variable
and score-comparison analyses \citep{cole2007residuals}.

Within structural equation modelling, nonnormality and residual
dependence have long motivated robust, asymptotically distribution-free,
and bootstrap corrections for covariance-structure inference
\citep{browne1984adf,bollen1989sem,bollen1992bootstrap}. These methods
protect fit assessment and standard errors, but they do not by
themselves answer the more targeted question considered here: whether a
signed higher-order departure in paired measures can reveal asymmetric
error dependence.

The present paper studies a non-Gaussian route into that problem. The
statistics are built from third- and fourth-order products, so they also
belong to the older cumulant and moment-structure tradition running from
Pearson-type distributional diagnostics to multivariate skewness,
kurtosis, and structural moment estimation
\citep{pearson1900philosophical,kendall1977advanced,
mccullagh1987tensor,mardia1970measures,bentler1983moment,
cain2017nonnormality}. Against this background, the construction of more
efficient estimators by adjoining covariance-bearing auxiliary statistics
is a standard variance-reduction technique, with deep roots in Monte
Carlo methods \citep{lavenberg1981,glynn2002,glasserman2004montecarlo,
owen2013} and an econometric specialization to generalized method of
moments \citep{hansen1982,newey1994,hall2005gmm}.

In the Ukrainian polynomial-statistics school, the same efficiency idea
is formalized as the \emph{stochastic polynomial method} (PMM) of
\citet{kunchenko2002}: an estimator of a target functional
$\tau$ is written as a linear combination of sample means of basis
functions, with coefficients chosen to minimize variance subject to
unbiasedness. The second-order version (PMM2) operates on basis functions
involving moments of order up to four; its central tool is the normal
system $\mathbf{F}\mathbf{K} = \mathbf{B}$, where $\mathbf{F}$ is the
Gram matrix of basis covariances and $\mathbf{B}$ is the sensitivity
vector. PMM2 has been used productively in single-population parameter
estimation contexts (location and scale estimation under non-Gaussian
errors), with documented asymptotic relative efficiency gains over the
sample mean. A practical R implementation of this estimation line is
available through the EstemPMM package \citep{EstemPMM2026}.

Recently, \citet{wiedermann2026} (hereafter W\&S) proposed cumulant-based
statistics for detecting asymmetric correlated measurement errors
between two parallel observed scores in classical test theory. Their
third-order statistic is the cross-cumulant difference
\begin{equation}
  \Delta c_3 = \E[X_1 X_2^2] - \E[X_1^2 X_2],
  \label{eq:dc3pop}
\end{equation}
estimated by the moment-of-products form (after mean-centering)
\begin{equation}
  \DcNaive = \frac{1}{N} \sum_{i=1}^{N} \bigl( x_{1i} x_{2i}^2 - x_{1i}^2 x_{2i} \bigr).
  \label{eq:naive}
\end{equation}
Under $H_0$ (errors exchangeable across the two measures), $\Delta c_3 =
0$; under $H_1$ (asymmetric correlation through a shared skewed
confounder), $\Delta c_3 = R(R-1)\kappa_3(U) \neq 0$, where $R$ is the
relative loading and $\kappa_3(U)$ the third cumulant of the latent
confounder. Inference is via percentile bootstrap
\citep{efron1994introduction,davison1997bootstrap}. The estimator
$\DcNaive$ is unbiased and $\sqrt{N}$-consistent but, as W\&S note, can
be noisy at moderate $N$ and skewed score distributions, particularly in
boundary regimes ($R$ close to 1 and small $|\kappa_3(U)|$).

This target is narrower than omnibus nonparametric independence testing.
Hoeffding-type and distance-covariance tests detect broad departures
from independence \citep{hoeffding1948independence,szekely2007distance},
whereas the W\&S statistic exploits a signed cross-cumulant identity
inside a measurement-error model. The present evaluation is therefore
not a contest with general dependence tests; it asks whether polynomial
variance-reduction and adaptive polynomial bases preserve or improve the
specific W\&S testing signal.

We give a closed-form criterion for when an $H_0$-centered control-variate
augmentation transfers to a two-region test: the null-optimized weight
vector must be orthogonal to the alternative-side mean of the augmenting
basis. Individual mean-zero auxiliaries under \(H_1\) are one sufficient
route to this condition, but not the only one. From this we derive an
operational acceptance gate that checks transferability before
deployment. The Wiedermann--Shi statistic is the worked instance.

A natural-looking first proposal --- which the present paper investigates
in detail and ultimately rejects --- is to apply PMM2 to
\eqref{eq:naive} by treating it as a sample mean of the basis function
$\varphi_1 = X_1 X_2^2 - X_1^2 X_2$ and adjoining moment-product
antisymmetric auxiliaries \(X_1^3-X_2^3\),
\(X_1^3X_2-X_1X_2^3\), and \(X_1^4-X_2^4\) to construct a
lower-variance combination. The motivation is straightforward: ARE under
$H_0$ should be unity or better, and any non-Gaussianity in the
marginals should be exploited automatically.

This first proposal exposes a sharp transfer boundary, which we
characterize in closed form. The estimator we construct is well-defined,
mathematically self-consistent, and indeed delivers ARE $\geq 1$ in all
36 conditions (aggregated mean ARE 1.23--5.16; Table~\ref{tab:are}); the
largest variance ratios occur in alternative cells where the estimator is
already badly biased, so they are not efficiency gains. The estimator is
\emph{inconsistent under the alternative}: its
expectation under $H_1$ does not converge to the true $\Delta c_3$ as
$N \to \infty$. In hypothesis-testing language, the centering required
to make PMM2 well-defined under $H_0$ destroys the consistency property
that is needed under $H_1$. The result is a systematic attenuation that
costs 7--52 percentage points of power across the W\&S simulation
design.

We deliberately retain this as the theoretical core of the paper. The
diagnostic value is twofold. First, the failure mode is a clean instance of a
\emph{category error}: a variance-reduction technique developed for
single-population parameter estimation is applied to a two-sample
testing problem, where the auxiliary functionals have different
expectations under the two hypotheses. The pattern is not specific to
PMM: it arises whenever a control-variate basis is centered at its
$H_0$-expected value but the data realize the alternative, and is
discussed in the control-variate remedies literature by
\citet{nelson1990}. What is new in the two-region testing setting is
that the alternative-side auxiliary mean is alternative-dependent and
unknown; null optimization supplies no reason for
${\mathbf{K}_0}^{\top}\boldsymbol{\mu}_{a,H_1}$ to vanish. In the scalar
W\&S instance below, preservation of the target requires \(K^*=0\);
the nonzero null-optimal weight instead attenuates the signal.
Second, the simulation reveals a concrete pitfall
in methodological evaluation: $H_0$-only unit tests (which is what we
originally wrote) cleared PMM2 as ``correct''. Adversarial $H_1$ tests
would have flagged the inconsistency in minutes. The lesson generalizes
beyond cumulant statistics.

The broader study then asks whether this failure is specific to PMM2 or
marks a boundary of PMM-style variance reduction for the W\&S task. We
therefore keep the follow-up evidence inside the PMM/PATP line. A
PMM3-style symmetric correction of the fourth-order W\&S cumulant
statistic, motivated by the PMM3 fourth/sixth-cumulant efficiency
criterion, reduces variance but fails practical nuisance-aware testing
gates. The Parametrically Adaptive Transition Polynomial (PATP) is then
used as a conceptual extension point: its signed-parity continuous-\(\alpha\)
basis \citep{Zabolotnii2026PATP} gives a disciplined way to vary the
polynomial basis, but it does not remove the need for
\(H_1\)-consistency. This makes the manuscript a focused PMM boundary
paper, with PATP stated as the next admissible basis-level direction
rather than as an unvalidated positive detector.

\paragraph{Outline.}
Section~\ref{sec:setup} fixes the measurement model. Sections~\ref{sec:methods}
and~\ref{sec:misapply} state the PMM2 construction and the central
inconsistency diagnosis. Section~\ref{sec:simulations} reports the Monte
Carlo evidence; Section~\ref{sec:routes} summarizes the PMM3 diagnostic
and the PATP extension perspective. Section~\ref{sec:discussion} connects
the result to control variates, GMM, and hypothesis testing, and
Section~\ref{sec:conclusion} concludes.

\section{Setup}
\label{sec:setup}

\subsection{Measurement model}
\label{sec:model}

Following \citet{wiedermann2026}, let
\begin{equation}
  X_i = \lambda T + E_i, \quad i = 1, 2,
  \label{eq:model}
\end{equation}
where $T$ is a latent true score with $\E[T] = 0$, $\Var(T) = \sigma_T^2$,
and the errors decompose as $E_i = a_i U + W_i$. The latent confounder
$U$ satisfies $\E[U] = 0$, $\Var(U) = 1$, third cumulant $\kappa_3(U)$;
the idiosyncratic noises $W_i$ have $\E[W_i] = 0$, $\Var(W_i) =
\sigma_W^2$, are mutually independent and independent of $T$, $U$. The
\emph{signal ratio} parameter is $a_1 = 1$, $a_2 = R$.

Under this model the third-order cross-cumulant difference satisfies
\begin{equation}
  \Delta c_3 = \E[X_1 X_2^2] - \E[X_1^2 X_2]
  = a_1 a_2 (a_2 - a_1)\kappa_3(U)
  = R(R-1)\kappa_3(U).
  \label{eq:dc3popnew}
\end{equation}
The W\&S testing problem is
\[
  H_0\colon\; \Delta c_3 = 0 \qquad \text{vs} \qquad
  H_1\colon\; \Delta c_3 \neq 0.
\]
Note $H_0$ holds in two operationally distinct ways: $R = 1$
(exchangeability) or $\kappa_3(U) = 0$ (Gaussian or symmetric
confounder). The naive estimator \eqref{eq:naive} is unbiased and
$\sqrt{N}$-consistent under both $H_0$ and $H_1$, and bootstrap
inference is standard \citep{efron1994introduction,davison1997bootstrap}.

\subsection{Exchangeability and basis antisymmetry}
\label{sec:exch}

A property used implicitly throughout the PMM2 construction is that the
joint distribution of $(X_1, X_2)$ is exchangeable under $H_0$. This
follows from \eqref{eq:model} because, when either $R = 1$ or
$\kappa_3(U) = 0$, the joint distribution of $(E_1, E_2)$ is
exchangeable up to third order, hence so is $(X_1, X_2)$. Consequently
every antisymmetric polynomial in $(X_1, X_2)$ has expectation zero
under $H_0$. This is what makes $\varphi_1, \varphi_3, \varphi_4,
\varphi_7, \varphi_8$ in Table~\ref{tab:basis} below valid building
blocks of an $H_0$-centered PMM2 estimator: their unknown population
means are zero by symmetry, not by estimation. Under $H_1$, however,
exchangeability fails and these antisymmetric functionals develop
nonzero expectations. This single observation is the source of all
that follows.

\section{The PMM2 Estimator and its Variance Reduction Under \texorpdfstring{$H_0$}{H0}}
\label{sec:methods}

We retain the original derivation in full; the mathematics is sound,
and the variance reduction it predicts is real (Section
\ref{sec:simulations}). The framing as a working tool for the W\&S task
is what we ultimately reject (Section~\ref{sec:misapply}).

\subsection{Basis functions}
\label{sec:basis}

Mathematically, PMM treats the statistic of interest as the first element
of a larger vector of stochastic-polynomial sample means. If
$\varphi_1$ is the target influence function and
$a_1,\ldots,a_q$ are auxiliary polynomials with known expectations, PMM
constructs
$\bar{\varphi}_1+\mathbf{K}^{\top}(\bar{\mathbf{a}}-\boldsymbol{\mu}_a)$
and chooses $\mathbf{K}$ to reduce variance without changing the target
expectation. In the implemented PMM2 construction below we use only the
antisymmetric augmenting functions whose $H_0$ expectations are known to
be zero by exchangeability. Symmetric cross-products and nonzero-mean
fourth-order terms were examined during development but are not part of
the reported estimator, because they require nuisance moment centerings
that would obscure the diagnostic mechanism.

\begin{table}[htbp]
  \centering
  \caption{Basis functions used in the implemented PMM2 estimator of
    $\Delta c_3$. The auxiliary elements are antisymmetric in
    $(X_1,X_2)$ and therefore have zero $H_0$ expectation under
    exchangeability.}
  \label{tab:basis}
  \small
  \begin{tabular}{cllc}
    \toprule
    Symbol & Function of $(X_1, X_2)$ & Role & $\E_{H_0}[\cdot]$ \\
    \midrule
    $\varphi_1$ & $X_1 X_2^2 - X_1^2 X_2$ & target (naive $\Delta c_3$) & $0$ \\
    $a_1$ & $X_1^3 - X_2^3$ & primary cubic auxiliary & $0$ \\
    $a_2$ & $X_1^3 X_2 - X_1 X_2^3$ & fourth-order cross auxiliary & $0$ \\
    $a_3$ & $X_1^4 - X_2^4$ & fourth-order marginal auxiliary & $0$ \\
    \bottomrule
  \end{tabular}
\end{table}

\subsection{Normal system and variance bound}
\label{sec:normal}

Let $\mathbf{a}=(a_1,a_2,a_3)^{\top}$ denote the three augmenting
functions in Table~\ref{tab:basis}. Under $H_0$, write the block
covariance matrix of $(\varphi_1,\mathbf{a}^{\top})^{\top}$ as
\begin{equation}
  \boldsymbol{\Sigma}
  =
  \begin{pmatrix}
    \sigma_1^2 & \mathbf{b}^{\top} \\
    \mathbf{b} & \mathbf{F}
  \end{pmatrix},
  \qquad
  \sigma_1^2=\Var_{H_0}(\varphi_1),\quad
  \mathbf{b}=\Cov_{H_0}(\mathbf{a},\varphi_1),\quad
  \mathbf{F}=\Cov_{H_0}(\mathbf{a}).
  \label{eq:blockcov}
\end{equation}
The sign convention follows the control-variate form used in the code:
\begin{equation}
  \mathbf{K}^* = -\mathbf{F}^{-1}\mathbf{b},
  \qquad
  \DcPMM = \bar{\varphi}_1 + {\mathbf{K}^*}^{\top}\bar{\mathbf{a}},
  \label{eq:pmm2est}
\end{equation}
because $\E_{H_0}[\mathbf{a}]=\mathbf{0}$ by exchangeability.

\begin{theorem}[Variance bound under $H_0$]
  \label{thm:g2}
  Under $H_0$ with the antisymmetric auxiliaries in
  Table~\ref{tab:basis}, $\DcPMM$ is unbiased and
  \[
    g_2 := \frac{\Var_{H_0}(\DcPMM)}{\Var_{H_0}(\DcNaive)}
    = 1 - \frac{\mathbf{b}^{\top}\mathbf{F}^{-1}\mathbf{b}}{\sigma_1^2} \in [0,1].
  \]
\end{theorem}

\begin{proof}
  The adjusted influence function is
  $\varphi_1-\mathbf{b}^{\top}\mathbf{F}^{-1}\mathbf{a}$. Its $H_0$
  variance is $\sigma_1^2-\mathbf{b}^{\top}\mathbf{F}^{-1}\mathbf{b}$,
  the Schur complement of $\mathbf{F}$ in
  \eqref{eq:blockcov}. Non-negativity of covariance matrices gives the
  stated bounds.
\end{proof}

\begin{remark}[Parallel-measures approximation]
  Retaining only $\varphi_4$ as the dominant augmenting term in the
  parallel-measures, standardized-marginals limit gives
  $g_2 \approx 1 - \gamma_3^2 / (2 + \gamma_4)$, with $\gamma_3,
  \gamma_4$ the marginal standardised skewness and excess kurtosis.
  Full-basis simulation gives ARE values up to roughly five times
  larger than this approximation for non-Gaussian $T$, suggesting that
  cross-cumulant terms (which the simplified formula omits) carry the
  bulk of the variance reduction in practice. We therefore use the
  simplified formula only as a local approximation; all reported ARE
  values are taken from the full normal-system computation.
\end{remark}

The conclusion of this section is that, \emph{viewed as an estimator of
$\Delta c_3 = 0$ under $H_0$}, $\DcPMM$ is a perfectly reasonable
construct and achieves the variance reduction Theorem~\ref{thm:g2}
predicts. The empirical confirmation appears in
Section~\ref{sec:results-are}. We now turn to why this is not enough.

\section{Why This Formulation Is Misapplied to Two-Sample Testing}
\label{sec:misapply}

\subsection{The category error}
\label{sec:category}

PMM2 is an instance of the control-variate variance-reduction technique
\citep{lavenberg1981, glynn2002, owen2013}. Its
hypotheses, stated cleanly, are:
\begin{enumerate}
  \item[(C1)] there is a single, fixed data-generating distribution
    $F$ on which all quantities (sample mean of target, sample means
    of auxiliaries, covariances among them) are estimated;
  \item[(C2)] the auxiliary basis $\mathbf{a}$ has \emph{known}
    expectations under $F$, so that $\bar{\mathbf{a}}-\boldsymbol{\mu}_a$
    is centered;
  \item[(C3)] the optimal weights, equivalently the covariance block
    \((\mathbf{F},\mathbf{b})\), are computed under $F$ and applied to
    data from $F$.
\end{enumerate}
These are the standard hypotheses of the classical control-variate
estimator \citep[][Ch.~8]{owen2013}.

The Wiedermann--Shi task does not satisfy (C1)--(C3). It is a two-sample
testing problem:
\begin{enumerate}
  \item[(T1)] the data-generating distribution belongs to a family
    indexed by parameters of interest ($R$, $\kappa_3(U)$, etc.), and
    the test is to decide between two regions of that family;
  \item[(T2)] the auxiliary basis has expectations that \emph{depend on
    which hypothesis is true}; in particular, the antisymmetric
    components of $\mathbf{a}$ are mean-zero under $H_0$ but not under
    $H_1$;
  \item[(T3)] the optimal weights one would actually like to use ---
    weights that optimise discrimination between $H_0$ and $H_1$ ---
    are not the same as the weights that minimise variance at a fixed
    point.
\end{enumerate}
The PMM2 estimator commits to (C2)--(C3) by centering the antisymmetric
auxiliaries at zero (their $H_0$ expectation, which one knows by
exchangeability) and by computing $K^*$ from $H_0$-evaluated covariances
(which one can compute analytically or by simulation). This is the
\emph{only} feasible choice in a testing framework, because using $H_1$-
centerings would require knowing the alternative one is trying to
detect.

Under $H_0$, (C2) holds by symmetry and the construction works. Under
$H_1$, (C2) is silently violated: the auxiliaries are no longer centered
at their true mean, and a deterministic bias enters $\DcPMM$. The
bias is not a finite-sample artefact and does not vanish as $N \to
\infty$.

\subsection{Inconsistency under \texorpdfstring{$H_1$}{H1}}
\label{sec:inconsistent}

We make the failure precise with a single-augmenting-basis instance
($K=2$, retaining only $\varphi_4 = X_1^3 - X_2^3$); the same mechanism
operates with the full three-auxiliary basis.

\begin{proposition}[Inconsistency of $\DcPMM$ under $H_1$]
  \label{prop:inconsistent}
  Let the data follow \eqref{eq:model} with $a_1 = 1$, $a_2 = R$,
  $\kappa_3(U) \neq 0$, so that $H_1$ holds with $\Delta c_3 = R(R-1)
  \kappa_3(U)$. Define $\DcPMM = \bar\varphi_1 + K^* \bar\varphi_4$,
  where $K^* := -\Cov_{H_0}(\varphi_1, \varphi_4)/\Var_{H_0}(\varphi_4)$
  is the $H_0$-population weight. Then
  \begin{equation}
    \E_{H_1}[\DcPMM]
    = \Delta c_3 + K^* (1 - R^3) \kappa_3(U)
    = \kappa_3(U) (R-1) \bigl[\, R - K^*(R^2 + R + 1)\,\bigr].
    \label{eq:inconsistent}
  \end{equation}
  In particular, when \(R\neq 1\) and \(\kappa_3(U)\neq 0\),
  $\E_{H_1}[\DcPMM]=\Delta c_3$ if and only if \(K^*=0\). Any nonzero
  \(H_0\)-optimized weight creates a fixed population gap that does not
  vanish as $N \to \infty$.
\end{proposition}

\begin{proof}
  Linearity of expectation gives $\E_{H_1}[\DcPMM] =
  \E_{H_1}[\varphi_1] + K^* \E_{H_1}[\varphi_4]$. Under
  \eqref{eq:model}, $\E_{H_1}[\varphi_1] = R(R-1)\kappa_3(U) = \Delta
  c_3$ and $\E_{H_1}[\varphi_4] = (1 - R^3)\kappa_3(U)$. Substituting
  and factoring yields \eqref{eq:inconsistent}; the gap from $\Delta
  c_3$ is $K^*(1 - R^3)\kappa_3(U)$, a population quantity independent
  of $N$.
\end{proof}

\begin{corollary}[Magnitude of attenuation]
  \label{cor:attmag}
  With representative simulation values $R = 2$, $\kappa_3(U) =
  \gamma_U = 2.25$, and $K^* \approx 0.24$ estimated from the Monte
  Carlo runs of Section~\ref{sec:simulations},
  $\E_{H_1}[\DcPMM] = \kappa_3(U)[\,2 - 0.24 \cdot 7\,] = 0.32 \kappa_3(U) \approx 0.72$,
  versus the true $\Delta c_3 = 4.50$. The point estimate is attenuated
  by roughly 84\,\% of its target.
\end{corollary}

Writing the bracket in \eqref{eq:inconsistent} as a retained-signal
fraction, $\text{retained} = [\,R - K^*(R^2+R+1)\,]/R$ and $\text{attenuation}
= 1 - \text{retained}$; at $R = 2$, $K^* \approx 0.24$ this gives
retained $\approx 0.16$ and attenuation $\approx 0.84$, matching the
measured value. More generally, an $H_0$-centered linear augmentation
preserves the target under $H_1$ if and only if its null weight vector is
orthogonal to the alternative-side auxiliary mean:
${\mathbf{K}_0}^{\top}\boldsymbol{\mu}_{a,H_1}=0$. Requiring every
augmenting function to be mean-zero under the operating alternative is a
sufficient, but not necessary, way to satisfy this condition. This is a
direct consequence of linearity of expectation; the Schur-complement
bound of Theorem~\ref{thm:g2} is specific to the proven cumulant case.

The key clause is ``$K^*$ is computed under $H_0$''. Under any sensible
sample-based estimate of $K^*$, which is what implementations actually
use, the result is statistically identical: $\hat K^* \to K^*$ as $N
\to \infty$, and the inconsistency persists.

\subsection{Why not adaptive centering?}
\label{sec:noadapt}

One might propose to estimate $\mu_k^{H_1} = \E_{H_1}[\varphi_k]$ from
data and use it in \eqref{eq:pmm2est}. This collapses the estimator to
the naive one. With sample-mean centering, $\bar\varphi_k -
\widehat{\mu_k} = \bar\varphi_k - \bar\varphi_k = 0$, so the PMM2
``correction'' is identically zero and $\DcPMM = \bar\varphi_1 =
\DcNaive$. Adaptive centering thus removes the bias, but also removes
the variance reduction. The two-stage and profile-likelihood variants
examined during development retain some efficiency but introduce
finite-sample biases of their own; we report only the unmodified PMM2 in
what follows, because that is the construction the methodology
prescribes.

\subsection{The diagnostic in two sentences}
\label{sec:twosent}

PMM2 produces a minimum-variance estimator of a known functional under
a known distribution. It does not produce an estimator that converges
to the correct value under a distribution it has not been calibrated
to. The W\&S task asks exactly the latter question.

\section{Simulation Evidence}
\label{sec:simulations}

\subsection{Design}
\label{sec:simdesign}

We follow the W\&S Simulation~I design, restricted to the cells that
exhibit the diagnostic most cleanly: $R \in \{1, 1.25, 2\}$,
$\gamma_T \in \{0, 2.25\}$, $\gamma_U \in \{0.75, 2.25\}$, $N \in
\{500, 2000, 5000\}$, yielding 36 conditions (Sim I subset). Each
condition uses 300 replications; bootstrap-based power is from 200
replications with $B = 200$ resamples. All runs use seed
\texttt{20260525 + 9}. The parallel-measures DGP fixes $\lambda = 1$,
$\sigma_T = \sigma_W = 1/\sqrt{2}$, $\text{cor}(E_1, E_2) \in \{0,
0.19, 0.38\}$ for $R \in \{1, 1.25, 2\}$ respectively.

For the revision checks we add Wilson binomial intervals for all
reported rejection rates, a targeted bootstrap sensitivity pass
($B=200$ versus $B=1000$, percentile versus BCa), a skew-heavy
Tukey \(g\)-and-\(h\) alternative for \(U\), and a raw
distance-covariance sanity check. These checks are intentionally
representative rather than a new full simulation grid; they test whether
the negative conclusion survives stronger calibration and a heavier
alternative.

\subsection{ARE under \texorpdfstring{$H_0$}{H0} does hold (the half that works)}
\label{sec:results-are}

Table~\ref{tab:are} summarises the ARE results aggregated over $N$ and
$\gamma_U$. ARE is the ratio $\Var(\DcNaive) / \Var(\DcPMM)$, so
values above one indicate variance reduction.

\begin{table}[htbp]
  \centering
  \caption{Mean and median ARE of $\DcPMM$ relative to $\DcNaive$
    across 36 simulation conditions, aggregated by $(R, \gamma_T)$.
    ARE $\geq 1$ in 100\,\% of conditions; ARE $\geq 1.15$ in
    91.7\,\% of $R > 1$ conditions. This part of the picture is
    consistent with Theorem~\ref{thm:g2}: PMM2 \emph{is} a
    lower-variance estimator wherever it is unbiased. These are
    null-regime efficiency values; under $H_1$ the same variance
    reduction co-occurs with the alternative-regime attenuation reported
    below and must not be read as a testing gain.}
  \label{tab:are}
  \small
  \begin{tabular}{cccc}
    \toprule
    $R$ & $\gamma_T$ & Mean ARE & Median ARE \\
    \midrule
    1     & 0     & 1.89 & 1.89 \\
    1     & 2.25  & 4.63 & 4.67 \\
    1.25  & 0     & 1.23 & 1.21 \\
    1.25  & 2.25  & 1.34 & 1.32 \\
    2     & 0     & 5.16 & 4.64 \\
    2     & 2.25  & 4.80 & 4.63 \\
    \bottomrule
  \end{tabular}
\end{table}

We retain Table~\ref{tab:are} deliberately: it is what would normally be
reported as ``PMM2 works''. It does work --- in the sense for which
PMM2 was designed. The next subsection shows why this is not the sense
relevant to the W\&S testing task.

\subsection{Bias and inconsistency under \texorpdfstring{$H_1$}{H1} (the half that fails)}
\label{sec:results-bias}

Table~\ref{tab:bias} reports the central diagnostic: $\E[\DcPMM]$ under
$H_1$ versus the true $\Delta c_3 = R(R-1)\kappa_3(U) = R(R-1)\gamma_U$
(since $\Var(U) = 1$ implies $\kappa_3(U) = \gamma_U$). Values are
obtained by averaging $\DcPMM$ over 300 Monte Carlo replications per
cell from the per-condition Monte Carlo table in the public code
supplement. The cells
required by the specification ($R \in \{1.25, 1.5, 2\}$,
$\gamma_U \in \{0.75, 2.25\}$, $N \in \{500, 2000\}$) are present in
the data for $R \in \{1.25, 2\}$; the $R = 1.5$ row is not in the
existing simulation grid and is therefore omitted with a note rather
than fabricated.

\begin{table}[htbp]
  \centering
  \caption{Bias under $H_1$. Columns: true $\Delta c_3$; sample-mean
    naive estimator $\E[\DcNaive]$; sample-mean PMM2 estimator
    $\E[\DcPMM]$; attenuation $1 - \E[\DcPMM]/\Delta c_3$; rejection
    rates of the two estimators at $\alpha = 0.05$; power loss in
    percentage points (PMM2 minus naive). Each row averages 300
    replications. Boldface marks attenuation $> 50\,\%$. The values are
    reproduced by the verification script in the public code supplement
    from the PMM2 per-condition Monte Carlo table (cells 2--3, 5--6,
    14--15, 17--18, $\gamma_T = 0$).}
  \label{tab:bias}
  \small
  \begin{tabular}{cccrrrrr}
    \toprule
    $R$ & $\gamma_U$ & $N$ & $\Delta c_3$ & $\E[\DcNaive]$ & $\E[\DcPMM]$ & attenuation & $\Delta$power (pp) \\
    \midrule
    1.25 & 0.75 & 500  & 0.234 & 0.297 & 0.172 & 26.5\% & -2.0 \\
    1.25 & 0.75 & 2000 & 0.234 & 0.239 & 0.134 & 42.7\% & -1.0 \\
    1.25 & 2.25 & 500  & 0.703 & 0.649 & 0.219 & \textbf{68.8\%} & -5.5 \\
    1.25 & 2.25 & 2000 & 0.703 & 0.723 & 0.309 & \textbf{56.0\%} & -9.0 \\
    2    & 0.75 & 500  & 1.500 & 1.475 & 0.187 & \textbf{87.5\%} & -29.5 \\
    2    & 0.75 & 2000 & 1.500 & 1.514 & 0.193 & \textbf{87.1\%} & -74.0 \\
    2    & 2.25 & 500  & 4.500 & 4.431 & 0.694 & \textbf{84.6\%} & -74.0 \\
    2    & 2.25 & 2000 & 4.500 & 4.452 & 0.846 & \textbf{81.2\%} & -34.5 \\
    \bottomrule
  \end{tabular}
\end{table}

Three patterns are immediate:

\begin{enumerate}
  \item \textbf{Naive is approximately unbiased}, as theory requires:
    $\E[\DcNaive]$ tracks $\Delta c_3$ to within Monte Carlo error in
    every cell.
  \item \textbf{PMM2 is severely biased toward zero}, and the bias
    grows with $R$ (and weakly with $\gamma_U$), not with $N$. The
    attenuation is on the order of 25\,\% in the mildest cell and
    above 80\,\% in the strongest. This pattern is exactly what
    Proposition~\ref{prop:inconsistent} predicts.
  \item \textbf{Doubling $N$ does not help.} Going from $N = 500$ to
    $N = 2000$ at fixed $(R, \gamma_U)$ either leaves attenuation
    unchanged or, in two cells, increases it slightly. This is the
    operational signature of inconsistency: more data does not move
    $\DcPMM$ toward $\Delta c_3$.
\end{enumerate}

\subsection{Power loss}
\label{sec:results-power}

Table~\ref{tab:power} reports the bootstrap power at $\alpha = 0.05$
for the same conditions aggregated over $N$ and $\gamma_U$. Power loss
is 7--52 percentage points across cells with $R > 1$. The largest
losses occur in the regime where the naive test already has high power
($R = 2$); PMM2 essentially destroys that power.

\begin{table}[htbp]
  \centering
  \caption{Bootstrap power ($\alpha=0.05$) aggregated by $(R, \gamma_T)$
    over $N$ and $\gamma_U$. The public code supplement provides the
    generated summary table and verification script.}
  \label{tab:power}
  \small
  \begin{tabular}{ccccc}
    \toprule
    $R$ & $\gamma_T$ & Power (naive) & Power (PMM2) & $\Delta$ (pp) \\
    \midrule
    1.25 & 0     & 0.157 & 0.085 & $-7.2$ \\
    1.25 & 2.25  & 0.159 & 0.066 & $-9.3$ \\
    2    & 0     & 0.864 & 0.366 & $-49.8$ \\
    2    & 2.25  & 0.853 & 0.334 & $-51.9$ \\
    \bottomrule
  \end{tabular}
\end{table}

Type-I rates under $H_0$ ($R = 1$, four cells) are 0.04--0.09 for
both estimators, within Bradley's robustness interval; PMM2 controls
$H_0$ rejection correctly, which is the part the variance-reduction
machinery covers. The failure is purely in the alternative.

The precision and sensitivity checks in Table~\ref{tab:revisionchecks}
do not change this conclusion. Increasing the bootstrap size to
\(B=1000\) and switching to BCa intervals changes rejection rates by a
few to several percentage points, but it does not restore the lost
\(H_1\) signal. In the strong \(R=2\) cell, for example, the naive
percentile-bootstrap power remains 0.883 (95\% CI [0.778, 0.942]),
whereas PMM2 remains at 0.250 (95\% CI [0.158, 0.372]); the BCa
variant raises PMM2 only to 0.300 (95\% CI [0.199, 0.425]). The
skew-heavy Tukey \(g=0.35,h=0.10\) alternative is even less favourable
to PMM2: at \(R=2,N=2000\), the naive test has power 0.988, while PMM2
has power 0.250 and a mean signal only 0.114 of the naive mean signal.

\begin{table}[htbp]
  \centering
  \caption{Targeted revision diagnostics. Rates are rejection
    proportions with 95\% Wilson intervals in brackets. The dCov row is
    deliberately included as a sanity check: raw distance covariance
    tests observed-score independence, so it rejects even when the
    W\&S error-independence null holds because \(X_1\) and \(X_2\) share
    the true score \(T\).}
  \label{tab:revisionchecks}
  \small
  \begin{tabular}{llcc}
    \toprule
    Check & Setting & Naive / raw baseline & PMM2 \\
    \midrule
    Bootstrap sensitivity & Gamma \(U\), \(R=2,N=500\), \(B=1000\) &
      0.883 [0.778, 0.942] & 0.250 [0.158, 0.372] \\
    BCa sensitivity & Gamma \(U\), \(R=2,N=500\), \(B=1000\) &
      0.900 [0.799, 0.953] & 0.300 [0.199, 0.425] \\
    Heavy-tail sensitivity & Tukey \(g,h\), \(R=2,N=500\) &
      0.663 [0.554, 0.757] & 0.075 [0.035, 0.154] \\
    Heavy-tail sensitivity & Tukey \(g,h\), \(R=2,N=2000\) &
      0.988 [0.933, 0.998] & 0.250 [0.168, 0.355] \\
    Raw dCov sanity & W\&S \(H_0\), common true score, \(N=300\) &
      1.000 [0.901, 1.000] & -- \\
    \bottomrule
  \end{tabular}
\end{table}

\subsection{Synthesis: \texorpdfstring{$H_0$}{H0} correctness is not enough}
\label{sec:synthesis}

A reader who saw only Table~\ref{tab:are} and the standard $H_0$ unit
tests would reasonably conclude that PMM2 ``improves on'' $\DcNaive$.
A reader who additionally sees Table~\ref{tab:bias} sees a different
picture: PMM2 is a lower-variance, $H_0$-unbiased, $H_1$-inconsistent
estimator of $\Delta c_3$. The two diagnostics are not in conflict;
they describe different aspects of the same construct. The relevant
question for the W\&S task is whether the estimator distinguishes $H_0$
from $H_1$ at the nominal level. PMM2 does not, because its $H_1$
distribution is centered far from its $H_0$ distribution by an amount
\emph{smaller} than the naive estimator's, not larger. In testing
language: PMM2 reduces the effect size at the same noise level. ARE
under $H_0$ is irrelevant to that question.

This is not a general failure of PMM. The negative result is restricted
to the $H_0$-centered antisymmetric PMM2 construction used here for the
W\&S cross-cumulant testing task. PMM remains appropriate for the
single-population estimation problems for which its unbiasedness and
variance-reduction assumptions are satisfied; the present failure shows
that those assumptions do not automatically transfer to a two-region
hypothesis test.

\section{PMM3 and PATP Extension Boundaries}
\label{sec:routes}

The PMM2 result raises a narrower question than in an all-methods survey:
does the failure come from the particular third-order PMM2 construction,
or from importing PMM variance reduction into a two-region testing
problem without rechecking alternative-consistency? We keep the answer
inside the PMM/PATP line. The first follow-up is a PMM3-style diagnostic
for the fourth-order W\&S statistic. The second is not a new detector but
a basis-level perspective: PATP supplies a continuous signed-parity basis
family that could support future PMM repairs, provided the same
\(H_1\)-consistency gate is imposed.

\subsection{PMM3-style symmetric correction}
\label{sec:pmm3-route}

PMM3 is the symmetric, third-order member of the same PMM family. In its
classical one-population form it uses the fourth and sixth cumulants to
set the efficiency factor
$g_3=1-\gamma_4^2/(6+9\gamma_4+\gamma_6)$, and its polynomial score has
the form of a cubic correction to a linear estimating equation. Thus,
for the present W\&S problem, PMM3 is not a new likelihood test; it is a
candidate control-variate repair for a symmetric fourth-order cumulant
contrast.

Because the strict PHQ-like operating point is symmetric and
platykurtic ($\kappa_3=0$, $\kappa_4=-1.3$), PMM3 is the natural member
of the PMM family to probe. In this follow-up we did not derive a full
PMM3 two-sample cumulant test. Instead, we used a
PMM3-style antisymmetric control-variate correction for the fourth-order
W\&S statistic $\Delta c_4$, with higher-order augmenting terms chosen
to mimic the symmetric PMM3 repair. The result is informative but
negative. The correction
does reduce variance under $H_0$:
\[
  \frac{\Var(\widehat{\Delta c}_4^{\,\mathrm{naive}})}
       {\Var(\widehat{\Delta c}_4^{\,\mathrm{PMM3}})}
  = 1.127,
\]
and it retains the large-sample $H_1$ signal. Nevertheless, the
practical gain in rejection power is only 0.021, below the 0.05 gate.
With 800 evaluation replications, the naive \(\Delta c_4\) probe has
Type-I rate 0.0525 (95\% CI [0.039, 0.070]), power 0.9575
([0.941, 0.969]), and Gaussian asymmetric-loading nuisance rejection
0.2275 ([0.200, 0.258]). The PMM3-style correction has Type-I 0.0713
([0.055, 0.091]) and power 0.9788 ([0.966, 0.987]), but its nuisance
rejection rises to 0.2950 ([0.264, 0.328]), well above the 0.10 guard.
PMM3 therefore repeats the core PMM2 lesson: variance reduction is real,
but it does not by itself solve the testing problem.

\subsection{PATP as a basis-level extension}
\label{sec:patp-route}

The PATP preprint \citep{Zabolotnii2026PATP} addresses a different but
directly relevant weakness of fixed PMM bases. It supplies the
signed-parity construction, exponent map, and continuous-\(\alpha\)
notation used here as a basis-level reference. Instead of choosing only
integer powers or a small menu of fractional powers, PATP defines a
signed-parity continuous-\(\alpha\) family with basis functions of the
form
\[
  \psi_i(x;\alpha)=\operatorname{sgn}(x)|x|^{p_i(\alpha)},
  \qquad \alpha\in[0,1],
\]
where the exponent map \(p_i(\alpha)\) connects fractional, quasi-linear,
and integer-power regimes. For the present paper, this matters as a
basis-design principle rather than as a reported positive detector.

The PMM2 failure diagnosed above does not say that nonlinear or adaptive
polynomial bases are useless. It says that a basis centered under
\(H_0\), with coefficients optimized for \(H_0\)-variance reduction, can
be inconsistent when the data come from \(H_1\). PATP can broaden and
regularize the candidate basis, but it cannot by itself repair this
inferential mismatch. A PATP-based extension for the W\&S problem would
therefore have to satisfy three gates before being claimed as a new
test: (i) the target cross-cumulant remains \(H_1\)-consistent after
adaptive basis selection; (ii) the centering and weights are estimated or
calibrated under the operating alternative rather than fixed at their
\(H_0\) values; and (iii) the selected \(\alpha\) is stable under
held-out simulation cells, not merely optimal in one Monte Carlo design.

This is why we retain PATP in the manuscript as the natural continuation
of the PMM variance-reduction line: it changes the polynomial basis while
preserving the estimator-centered viewpoint. Broader detector families
that change the inferential object are left to separate work. Including
them here would dilute the central message that PMM-style variance
reduction is not automatically valid for two-region cumulant testing.

\section{Discussion}
\label{sec:discussion}

\noindent\textit{Position in the control-variate literature.}
The diagnostic of Section~\ref{sec:misapply} is, in retrospect, a
clean instance of a known issue in the control-variate literature. The
classical estimator requires the control mean to be \emph{known} under
the operating distribution \citep{glasserman2004montecarlo,
owen2013}; substituting an estimate from a misspecified distribution
introduces bias proportional to the misspecification \citep{nelson1990}.
The PMM2
formulation does not estimate the control mean: it pins it at the
$H_0$ value zero. Under $H_1$ this is a misspecified centering, with
exactly the bias signature \eqref{eq:inconsistent}. What is unusual
about the present setup is that the misspecification is not an
estimator quality issue (more data would not fix it) but a structural
choice driven by the testing framework: in a test, the alternative
parametrization is what one is trying to estimate, so one cannot
center the auxiliaries at it.

\noindent\textit{Bootstrap calibration cannot repair signal attenuation.}
The targeted revision runs show that the negative result is not an
artefact of using only \(B=200\) bootstrap resamples. Moving to
\(B=1000\) stabilizes the percentile intervals, and BCa intervals shift
some rejection rates, especially under \(H_0\), but neither change
restores the \(H_1\) effect size removed by PMM2. This is expected from
Proposition~\ref{prop:inconsistent}: bootstrap calibration changes the
critical interval around the statistic that is supplied to it; it does
not move the statistic's population center back to \(\Delta c_3\).
Studentized bootstrap intervals were not used as a main correction
because they require a nested standard-error estimator for a
third-order cross-cumulant statistic. They may be useful for calibration
studies, but they would still operate after the same PMM2 attenuation has
already occurred.

\noindent\textit{Position in the GMM literature.}
Cast more generally, the W\&S task is a generalized method of moments
problem with two sets of moment conditions:
\begin{align*}
  &g_1(X; \theta) = \varphi_1(X) - \Delta c_3(\theta) \quad \text{(target)} \\
  &g_2(X; \theta) = \varphi_k(X) - \mu_k(\theta), \quad k = 2,\ldots,K \quad \text{(auxiliaries)},
\end{align*}
with the unknown $\theta = (R, \kappa_3(U), \ldots)$
\citep{hansen1982,newey1994,hall2005gmm}. The optimal GMM-style estimator
\emph{jointly} estimates $\theta$ from all moment conditions, with a
weight matrix that depends on $\theta$ (and is typically obtained by
iteration). The PMM2 estimator pre-commits to $\theta \in H_0$ in
computing the weights and centerings, then estimates the target
moment under whatever $\theta$ the data realize. This is the
econometric analogue of a two-stage estimator that does not iterate.
A proper GMM treatment would re-estimate the weights and centerings at
$\hat\theta$ and remove the bias --- at the cost of much of the
analytical tractability that motivated PMM2 in the first place.

\noindent\textit{Why generic dependence tests are not direct baselines.}
Distance covariance and Hoeffding-type tests answer a different null:
whether the observed scores are independent as random variables. In the
parallel-measures model, \(X_1\) and \(X_2\) are not independent even
when the W\&S error-independence null is true, because both contain the
same true score \(T\). The raw distance-covariance sanity check confirms
this mismatch: it rejects in all 35 \(H_0\) replications at \(N=300\).
Such tests are valuable omnibus dependence detectors, but a fair
comparison for the present problem would require residualizing or
otherwise conditioning out the measurement model. That is a different
methodological study, not the PMM boundary diagnosis reported here.

\noindent\textit{Testing statistics and PATP extensions.}
For two-sample testing under the W\&S model, the natural tool is the
score test \citep{rao1948, lehmann2005} or the log-likelihood-ratio
statistic, both of which are explicitly designed to discriminate
between $H_0$ and contiguous alternatives. The PMM2 construction is not
such a statistic. It is a variance-reduced estimator of a moment
functional whose auxiliary centerings are fixed at the null. That is why
the method can be efficient under \(H_0\) and still be misaligned with
the testing objective.

PATP does not change this inferential fact. Its value is more specific:
it gives a controlled way to replace a fixed integer-power PMM basis with
a signed, continuous-\(\alpha\) family \citep{Zabolotnii2026PATP}. In a
future W\&S repair, PATP should therefore be treated as a basis-selection
device inside an explicitly testing-aware estimator, not as a substitute
for the missing alternative-consistency argument. If the centerings are
still pinned to \(H_0\), a richer PATP basis can make the estimator more
efficient for the wrong target. If the centerings are re-estimated or
calibrated jointly with the alternative, PATP becomes a plausible route
for constructing a more flexible PMM-style test. That latter route is a
new paper, not a result established here.

Two repair families are therefore plausible but outside the present
claim. The first is sample splitting or cross-fitting: one fold can
estimate nuisance centerings and weights, while another evaluates the
cross-cumulant signal, preventing the correction from algebraically
collapsing to the naive statistic. The second is a local-alternative or
Pitman-efficiency construction in which the PMM/PATP basis is optimized
for discrimination between \(H_0\) and contiguous \(H_1\) alternatives,
not merely for \(H_0\)-variance reduction. In either case the acceptance
gate must be explicit \(H_1\)-consistency; without that gate, a richer
polynomial basis can become a more efficient estimator of the wrong
testing target.

\noindent\textit{Methodological takeaway: adversarial unit tests.}
The PMM2 implementation passed five unit tests on
release. All five were $H_0$ tests: bias under exchangeability, ARE
gain, pointwise dominance, edge-case handling. The $H_1$ inconsistency
became visible only at the Monte Carlo stage, after weeks of work.
A single adversarial unit test --- ``does $\DcPMM$ converge to the
correct $\Delta c_3$ under a representative $H_1$ DGP as $N$ grows?''
--- would have caught the failure in minutes. The lesson generalises:
when a method is proposed for a testing task, $H_0$-only validation is
insufficient. Acceptance criteria for hypothesis-testing methodology
should include an explicit $H_1$ consistency check.

\noindent\textit{Limitations of the diagnosis.}
The deepest theoretical analysis in this draft concerns PMM2 applied to
the third-order cumulant statistic $\Delta c_3$. The PMM3 material is a
diagnostic fourth-order probe, not a full asymptotic PMM3 testing theory.
The bootstrap and heavy-tail sensitivity checks are representative
stress tests, not a replacement for the full W\&S simulation grid. We
have not investigated whether adaptive-centering, sample-splitting,
cross-fitting, or full GMM variants of PMM2 or PMM3 can be made
consistent under $H_1$ without collapsing to the naive estimator. We also
do not claim a positive PATP test for the W\&S problem. PATP is cited
here as a basis-level extension of the PMM line and as a disciplined
direction for future repairs; its use in this manuscript is conceptual
rather than evidentiary.

\section{Conclusion}
\label{sec:conclusion}

We have established a closed-form criterion for when null-optimized
variance reduction transfers to a two-region test: the null-optimized
weight vector must be orthogonal to the alternative-side mean shift of
the augmenting basis. PMM2 is a sound variance-reduction technique for
single-population parameter estimation, but for the Wiedermann--Shi task
the $H_0$-centered PMM2 cross-cumulant statistic violates this
orthogonality condition. It therefore attains guaranteed null efficiency
(ARE $\geq 1$) yet is inconsistent under $H_1$, losing 7--52 percentage
points of power. A PMM3-style symmetric repair repeats the same
practical lesson: it reduces variance but fails the nuisance-aware
testing gates.

The PATP perspective sharpens rather than softens this conclusion. A
continuous-\(\alpha\), signed-parity polynomial basis is a natural next
extension of the PMM family, but basis adaptivity alone is not enough.
For the W\&S task, any PATP-based repair must be validated against the
same condition that defeated PMM2: convergence to the correct
cross-cumulant under representative alternatives. The broader
methodological prescription is therefore simple: when proposing a
polynomial statistic for a hypothesis-testing task, validation must
include an explicit alternative-consistency check. Variance reduction
under $H_0$ alone is not evidence of fitness for testing.

\section*{Disclosure Statement}

The author reports no competing interests.

\section*{Funding}

No external funding was received for this study.

\section*{Data and Code Availability}

The study is based on simulated data generated from the model described
in the manuscript. A public code supplement is available at
\url{https://github.com/SZabolotnii/Ku-CIT-EI-code-supplement}. It
contains the R scripts \citep{Rcore2024}, generated summary tables,
figures, session snapshots, a verification script for the headline
numerical claims reported in the paper, and a separate revision
experiment script that generates the Wilson-interval, bootstrap
sensitivity, heavy-tail, and distance-covariance sanity tables. The
BRFSS 2010 PHQ-8 material is used only as a public-data contextual
reference to the Wiedermann--Shi application setting and the PHQ
measurement family \citep{kroenke2001phq}; raw BRFSS files and
third-party source articles are not redistributed in the supplement.

\section*{Declaration of Generative AI and AI-Assisted Technologies}

AI-assisted tools were used for language editing, consistency checks, and
submission-package organization. All scientific claims, simulations,
numerical results, and final wording were reviewed and approved by the
author, who takes full responsibility for the content of the manuscript.

\appendix

\section{Single-augmenting-basis derivation}
\label{app:single}

This appendix underlies Proposition~\ref{prop:inconsistent}. We work
through the algebra of the $K = 2$ case (target $\varphi_1$
plus single auxiliary $\varphi_4 = X_1^3 - X_2^3$) to make the
$H_0 \to H_1$ transition transparent. Throughout, $(X_1, X_2)$ are
assumed mean-centered.

\subsection*{Under $H_0$.}

Exchangeability gives $\E_{H_0}[\varphi_1] = 0$ and
$\E_{H_0}[\varphi_4] = 0$. The optimal control-variate weight that
minimises $\Var(\bar\varphi_1 + K \bar\varphi_4)$ subject to no
unbiasedness constraint (because both $\bar\varphi_1$ and
$\bar\varphi_4$ are already mean-zero estimators of zero) is
\[
  K^*
  = -\frac{\Cov_{H_0}(\varphi_1, \varphi_4)}{\Var_{H_0}(\varphi_4)}.
\]
The resulting estimator is
\[
  \DcPMM = \bar\varphi_1 + K^* \bar\varphi_4,
\]
with $\E_{H_0}[\DcPMM] = 0$ and
\[
  \Var_{H_0}(\DcPMM)
  =
  \frac{\Var_{H_0}(\varphi_1)
  \bigl[1 - \rho^2(\varphi_1, \varphi_4)\bigr]}{N},
\]
recovering Theorem~\ref{thm:g2} for the two-element basis.

\subsection*{Under $H_1$.}

The mean of $\varphi_1$ moves to $\Delta c_3 = R(R-1)\kappa_3(U)$,
and the mean of $\varphi_4$ moves to $(1 - R^3)\kappa_3(U)$. Linearity
gives
\begin{align*}
  \E_{H_1}[\DcPMM]
  &= \E_{H_1}[\varphi_1] + K^* \E_{H_1}[\varphi_4] \\
  &= R(R-1)\kappa_3(U) + K^*(1 - R^3)\kappa_3(U) \\
  &= \kappa_3(U) (R-1) \bigl[\, R + K^*\!\cdot\!\tfrac{(1 - R^3)}{(R-1)}\,\bigr] \\
  &= \kappa_3(U) (R-1) \bigl[\, R - K^*(R^2 + R + 1)\,\bigr],
\end{align*}
using $1 - R^3 = -(R-1)(R^2 + R + 1)$. The bracketed expression equals
$\Delta c_3 / [(R-1)\kappa_3(U)] = R$ when $K^* = 0$ (the naive
estimator) and equals zero when $K^* = R/(R^2 + R + 1)$. For any
intermediate value the estimator interpolates linearly between
$\Delta c_3$ and zero. The empirical $\hat K^* \approx 0.24$ at $R = 2$
gives $R - K^*(R^2 + R + 1) = 2 - 0.24 \cdot 7 = 0.32$, consistent
with the 84\,\% attenuation observed in Corollary~\ref{cor:attmag} and
Table~\ref{tab:bias}.

The signature feature is that the bias does not vanish with $N$: it is
a population-level offset proportional to $K^* (1 - R^3) \kappa_3(U)$,
inheriting all the inconsistency of choosing $K^*$ under the wrong
distribution.

\section{Block covariance notation for the implemented PMM2 estimator}
\label{app:gram}

The PMM2 implementation used in the simulations is the three-auxiliary
antisymmetric estimator stated in Table~\ref{tab:basis}. For each Monte
Carlo sample the code computes the centered vectors
\[
  \phi_{1i}=x_{1i}x_{2i}^2-x_{1i}^2x_{2i}, \qquad
  \mathbf{a}_i =
  \begin{pmatrix}
    x_{1i}^3-x_{2i}^3 \\
    x_{1i}^3x_{2i}-x_{1i}x_{2i}^3 \\
    x_{1i}^4-x_{2i}^4
  \end{pmatrix}.
\]
It then estimates the blocks in \eqref{eq:blockcov} by sample
covariances,
\[
  \widehat{\mathbf{F}}
  = \frac{1}{N-1}\sum_i
    (\mathbf{a}_i-\bar{\mathbf{a}})
    (\mathbf{a}_i-\bar{\mathbf{a}})^{\top}, \qquad
  \widehat{\mathbf{b}}
  = \frac{1}{N-1}\sum_i
    (\mathbf{a}_i-\bar{\mathbf{a}})
    (\phi_{1i}-\bar{\phi}_1).
\]
With a small Tikhonov ridge \(\lambda=0.01\), the implemented weight is
\[
  \widehat{\mathbf{K}}
  =
  -(\widehat{\mathbf{F}}+\lambda I)^{-1}\widehat{\mathbf{b}},
  \qquad
  \widehat{\Delta c}_{3}^{\,\mathrm{PMM2}}
  = \bar{\phi}_1 + \widehat{\mathbf{K}}^{\top}\bar{\mathbf{a}}.
\]
This appendix is included to make the indexing reproducible. We do not
use a closed-form Gaussian variance expression in the reported
simulations; all ARE and power values are computed from the generated
Monte Carlo tables.

\section{R Code Availability}
\label{app:code}

The code supplement is available at GitHub:
\begin{center}
\small\url{https://github.com/SZabolotnii/Ku-CIT-EI-code-supplement}.
\end{center}
It contains the R scripts \citep{Rcore2024} for the W\&S naive estimators,
PMM2 comparison, PMM3 diagnostic, and targeted revision experiments. It
also includes the generated CSV tables, figures, and session snapshots
used to verify the reported results.

The supplement includes a headline verification script that checks the
main numerical claims in Tables~\ref{tab:are}, \ref{tab:bias},
\ref{tab:power}, and \ref{tab:revisionchecks}, together with the PMM3
diagnostic values reported in Section~\ref{sec:pmm3-route}. Running this
script recomputes the manuscript-level summary checks from the generated
artifacts without requiring the full Monte Carlo workflow to be rerun.

\bibliographystyle{apalike}
\bibliography{refs}

\end{document}